\documentclass[11pt]{article}

\usepackage[top=4.4cm, bottom=4.4cm, left=3.4cm, right=3.4cm, letterpaper ]{geometry}

\usepackage{amsmath,amssymb,amsthm,amscd}
\usepackage{graphicx}
\usepackage{psfrag}
\usepackage[active]{srcltx}

\newcommand{\C}{\mathbb{C}}
\newcommand{\B}{{\cal B}}

\newcommand{\pr}{{\rm pr}}

\newcommand{\A}{{\cal A}}

\newcommand{\D}{{\cal D}}
\renewcommand{\S}{{\cal S}}

\newcommand{\h}{{\cal H}}
\newcommand{\scal}[2]{\langle #1| #2\rangle}

\newcommand{\id}{{\rm id}}

\newcommand{\ot}{\otimes}
\newcommand{\la}{\lambda}
\DeclareMathOperator{\tr}{{\rm tr}}

\DeclareFontFamily{U}{mathx}{\hyphenchar\font45}
\DeclareFontShape{U}{mathx}{m}{n}{
      <5> <6> <7> <8> <9> <10>
      <10.95> <12> <14.4> <17.28> <20.74> <24.88>
      mathx10
      }{}
\DeclareSymbolFont{mathx}{U}{mathx}{m}{n}
\DeclareMathSymbol{\bigtimes}{1}{mathx}{"91}

\newcounter{mnotecount}[section]

\newtheorem{thr}{Theorem}
\newtheorem{lm}[thr]{Lemma}
\newtheorem{df}[thr]{Definition}

\numberwithin{equation}{section}
\numberwithin{thr}{section}

\begin{document}
\title{A modification of the projective construction of quantum states for field theories\footnote{This is an author-created copyedited version of an article accepted for publication in Journal of Mathematical Physics. The definitive publisher authenticated version is available online at  http://dx.doi.org/10.1063/1.4989550.}}
\author{Jerzy Kijowski$^1$, Andrzej Oko{\l}\'ow$^2$ }
\date{June 30, 2017}

\maketitle
\begin{center}
{\it 1. Center for Theoretical Physics of the Polish Academy of Sciences\\
Al. Lotnikow 32/46, 02-668 Warsaw, Poland\smallskip\\
kijowski@cft.edu.pl}\medskip\\

{\it  2. Institute of Theoretical Physics, Warsaw University\\ ul. Pasteura 5, 02-093 Warsaw, Poland\smallskip\\
oko@fuw.edu.pl}
\end{center}
\medskip

\begin{abstract}
The projective construction of quantum states for field theories may be flawed---in some cases the construction may possibly lead to spaces of quantum states which are ``too small'' to be used in quantization of field theories. Here we present a slight modification of the construction which is free from this flaw. 
\end{abstract}

\section{Introduction}

\subsection{Projective construction of quantum states and its possible flaw}

The projective construction of quantum states for field theories reads as follows \cite{kpt}. The point of departure for it is a phase space of a field theory. Usually this space is infinite dimensional. By an appropriate choice of a finite number of degrees of freedom (d.o.f.) from those constituting the phase space, one obtains a finite physical system $\lambda$. This system can be ``quantized'' by assigning to it a Hilbert space $\h_\lambda$ representing pure quantum states of the system and the space $\D_\lambda$ of all density operators on $\h_\la$ representing mixed quantum states. In order to obtain a space of quantum states for the field theory, one needs to define a family $\Lambda$ of finite physical systems which satisfies the following properties. Firstly, the physical systems altogether should encompass all the d.o.f. of the original phase space. Secondly, it should be possible to organize the systems into a {\em directed set} such that a system $\la'$ is bigger than or equal to a system $\la$, $\la'\geq\la$, if the system $\la$ is a {\it subsystem} of $\la'$. Thirdly, given system $\la'$ and its subsystem $\la$, it should be possible to project the mixed states of $\la'$ onto those of $\la$ (i.e. the elements of $\D_{\la'}$ onto the ones of $\D_\la$) via a partial trace $\pi_{\la\la'}$, which ``annihilates'' those quantum d.o.f. of $\la'$ that are lacking in $\la$. Finally, the set $\{\D_\la,\pi_{\la\la'}\}_{\la\in\Lambda}$ is required to be a {\em projective family}. If the family $\Lambda$ of physical systems satisfies all these requirements then one defines the space of quantum states for the field theory as the {\em projective limit} $\D$ of the projective family $\{\D_\la,\pi_{\la\la'}\}_{\la\in\Lambda}$. 

The idea of this construction is closely related to the fact that in every real experiment one can measure only a finite number of observables---every system $\la$ is meant to be built from d.o.f. which can be measured in an experiment. This fact excludes in practice the full knowledge of a state $\rho\in \D$ being a special net $(\rho_\la) $ of states such that $\rho_\la\in\D_\la$ since $\rho$ encompasses information about an infinite number of d.o.f.---what we can really know is only a state $\rho_\la$ for some $\la$ and therefore such a state should be seen as an {\it approximation} of the state $\rho$ \cite{kpt}. This means that when designing and carrying out an experiment and when analyzing its results one would actually work with a finite system $\la$ using quantum states in $\D_\la$ and quantum observables of the system. Thus the whole projective construction serves two purposes $(i)$ it makes it possible to compare two experiments corresponding to different finite systems---if one experiment corresponds to $\la$ and the other to $\la'$, then there exists a finite system $\la''\geq\la,\la'$ and this fact allows us to interpret the two original experiments as ones measuring observables of the same finite system $\la''$ \cite{proj-lt-I}, $(ii)$ it gives a precise meaning to the statement that a state $\rho_\la$ is an approximation of a state $\rho$ of the full quantum field theory.

In the original paper \cite{kpt} this idea was applied to linear phase spaces. Further development of the projective construction presented in \cite{non-comp,q-nonl,q-stat,proj-lt-II,proj-lqg-I} consisted in applying this idea to more and more general phase spaces including finally the one underlying Loop Quantum Gravity (LQG).

However, as realized in \cite{sl-phd} the projective construction is in general uncertain because of a troublesome feature of the very projective limit---the projective limit of a projective family may actually be the empty set. Let us emphasize the fact that all sets $\{\D_\la\}$ being non-empty and all partial traces $\{\pi_{\la\la'}\}$ being surjective does not guarantee that the corresponding projective limit $\D$ is non-empty, since there are known projective families with these properties which possess empty limits \cite{empty-0,empty}. Moreover, to use the limit $\D$ as the space of quantum states for a field theory one should not only prove that it is non-empty but also that it is ``large enough'' to describe all quantum d.o.f. of the theory. 

We do not yet know any projective family $\{\D_\la,\pi_{\la\la'}\}_{\la\in\Lambda}$, constructed according to the above prescription, which leads to ``too small'' a projective limit, but the lack of certainty that every space of quantum states provided by the projective construction is suitable for quantization is unsatisfactory since every quantum field theory based on an uncertain state space is also uncertain. 

Thus it is desirable to remove the possibility that spaces of projective quantum states may be ``too small''. So far two partial solutions to this problem are known---we will describe them in the next section and show that there are relevant applications of the projective construction for which neither of these solutions works. 

The goal of this paper is to present a general solution to the problem of ``too small'' spaces of projective quantum states. The solution will consist in a suitable modification of the hitherto construction---we will show that {\em every} family $\{\D_\la,\pi_{\la\la'}\}_{\la\in\Lambda}$ can be naturally extended to a projective family such that its projective limit is not only non-empty but is also ``sufficiently large'' to serve as the space of quantum states for the field theory for which the family $\{\D_\la,\pi_{\la\la'}\}_{\la\in\Lambda}$ is constructed.

\subsection{Partial solutions}

Suzanne  Lan\'ery and Thomas Thiemann showed in \cite{proj-lt-IV} that if a directed set $(\Lambda,\geq)$ of finite physical systems admits a {\em countable cofinal subset}\footnote{A directed subset $(\Lambda',\geq)$ of a directed set $(\Lambda,\geq)$ is called {\em cofinal} if for every $\la\in\Lambda$ there exists $\la'\in\Lambda'$ such that $\la'\geq\la$.}, then the corresponding projective limit $\D$ is ``large enough''.  

However, a directed set $(\Lambda,\geq)$ which does not admit any countable cofinal subset is nothing out of the ordinary: if one constructs projective quantum states for any {\em diffeomorphism invariant (background independent)} field theory with local d.o.f., then it is natural to use a set of this sort. The reason for this is that it would be very difficult to build projective quantum states on the basis of diffeomorphism invariant d.o.f. and therefore one has to use diffeomorphism covariant ones. And if one accepts {\em one} finite physical system $\la_0$ defined by such d.o.f., then one consequently accepts {\em all} systems obtained from $\la_0$  by the action of all diffeomorphisms otherwise there would be a risk of breaking diffeomorphism symmetry. Thus one ends up with a set $\Lambda$ which is uncountable. On the other hand since the relation $\geq$ is a relation system--subsystem, every cofinal subset of $\Lambda$ has to encompass all d.o.f. encompassed by $\Lambda$. But the set of finite systems obtained from $\la_0$ by the action of all diffeomorphisms encompasses an uncountable number of d.o.f. and therefore the set $\Lambda$ cannot possess any countable cofinal subset.

Restricting ourselves to physically relevant applications of the projective construction, we conclude that sets of finite systems which do not admit countable cofinal subsets appear in applications to $(i)$ background independent quantization of general relativity (GR) (vacuum or coupled to matter fields) and to $(ii)$ quantization of background independent toy-models (with local d.o.f.), which may be helpful in solving problems\footnote{A very important and difficult problem is the one of solving constraints on spaces of projective quantum states constructed for GR, and it is hard to imagine that this problem could be solved without the help of toy-models} encountered in quantization of GR. To be more precise: such sets are used in the following existing applications of the projective construction:
\begin{enumerate}
\item the application \cite{proj-lqg-I} by Lan\'ery and Thiemann to LQG;  
\item the application \cite{q-stat} to the Teleparallel Equivalent of General Relativity and a related toy-model \cite{os}.
\item the application  \cite{lqg-tens} to theories of tensor fields coupled to LQG;
\item the three different applications \cite{non-comp,q-nonl,constr-dpg} to a toy-model called  degenerate Pleba\'nski gravity. 
\end{enumerate}
Let us also note that the set of finite systems applied in the original construction \cite{kpt} for scalar field theories does not admit any countable cofinal subset since it is built from subsets of a three-dimensional space. Of course, if one used this construction to quantization of a scalar field on a fixed background then the set of finite systems could be restricted to a countable one, but in the case of quantization of GR coupled to a scalar field it would be necessary to keep the uncountable set intact.  

Regarding future applications in which this sort of sets of finite systems would be used, let us mention only one, namely a construction of projective quantum states for LQG expressed in terms of the complex Ashtekar variables \cite{a-var-1,a-var-2}. The gauge group of these variables is $SL(2,\C)$ and because of its non-compactness no one so far has been able to construct any acceptable quantum state space for this formulation of GR. The Hilbert space used in LQG was obtained in \cite{al-hoop} by breaking the $SL(2,\C)$ symmetry to the $SU(2)$ one (which amounts to breaking the Lorentz symmetry of GR to that of three-dimensional rotations) and in \cite{proj-lqg-I} Lan\'ery and Thiemann constructed their projective quantum states also for LQG with the $SU(2)$ symmetry. A wish to construct projective quantum states for LQG with the $SL(2,\C)$ symmetry was the main motivation for the paper \cite{non-comp}. It seems now that the Lan\'ery-Thiemann construction for LQG with the $SU(2)$ symmetry should work for LQG with the $SL(2,\C)$ symmetry as well and yield the desired space of projective states for the latter model. Of course, this space may be possibly empty or ``too small''. 
     
The second partial solution to the problem of ``too small'' spaces of projective quantum states was presented by Lan\'ery and Thiemann in \cite{proj-lqg-I}---they proved that if a projective family $\{\D_\la,\pi_{\la\la'}\}_{\la\in\Lambda}$ is constructed on the basis of so called holonomy-flux algebra for a theory of connections with a {\em compact} structure group, then its projective limit is ``sufficiently large''. However, the only theory of this sort among those mentioned above is LQG with the $SU(2)$ symmetry. 

Let us finally mention a potential general solution to the problem proposed by Lan\'ery and Thiemann in \cite{proj-lt-IV}. It consists in replacing an uncountable set of finite physical systems by a countable one of special properties: the countable set is required to be ``cofinal up to small deformations'' with respect to the original uncountable one.  Moreover, in the case of a diffeomorphism invariant theory every diffeomorphism has to be ``well approximated'' by automorphisms of the countable set, otherwise the diffeomorphism symmetry would be broken. These two requirements seem to be indispensable but on the other hand they make it very difficult to find any acceptable countable set of finite systems for any diffeomorphism invariant theory---Lan\'ery and Thiemann only managed to construct an example for a theory defined on an interval of the real line.

\section{Preliminaries}

\subsection{Projective limit}

Here we recall the notions of projective family and its projective limit following \cite{prof-gr}.

Let $(I,\geq)$ be a directed set. Consider a family $\{X_i,\varphi_{ii'}\}_{i\in I}$, where for every $i\in I$ $X_i$ is a set and $\{\varphi_{ii'}:X_{i'}\to X_i\}$ are maps defined for every $i',i\in I$ such that $i'\geq i$. If for every $i'',i',i\in I$ such that $i''\geq i'\geq i$
\[
\varphi_{ii'}\circ\varphi_{i'i''}=\varphi_{ii''},
\]
then $\{X_i,\varphi_{ii'}\}_{i\in I}$ is called a {\em projective family} (or an {\em inverse system}).

Suppose that $\{X_i,\varphi_{ii'}\}_{i\in I}$ is a projective family. Let $Y$ be a set equipped with a set of maps $\{\theta_i:Y\to X_i \}_{i\in I}$. We say that the maps are {\em compatible} with the projective family if for every $i'\geq i$
\[
\theta_i=\varphi_{ii'}\circ\theta_{i'}.
\]

Consider a set $X$ equipped with maps $\{\varphi_i:X\to X_i\}_{i\in I}$ compatible with a projective family $\{X_i,\varphi_{ii'}\}_{i\in I}$. If for every set $Y$ equipped with maps $\{\theta_i:Y\to X_i\}_{i\in I}$ compatible with the  family there exists a unique map $\theta:Y\to X$ such that
\[
\theta_i=\varphi_i\circ\theta,
\]
then $X$ is called a {\em projective limit} (or an {\em inverse limit}) of the projective family.

One can show  that for every projective family there exists a unique projective limit. Given projective family $\{X_i,\varphi_{ii'}\}_{i\in I}$, its limit $X$ can be described as follows:
\[
X=\{\ (x_i)\in \bigtimes_{i\in I} X_i \ | \ x_i=\varphi_{ii'}(x_{i'}) \ \text{for every $i'\geq i$} \ \}
\]
i.e. the limit consists of nets $(x_i)$ {\em compatible} with the maps $\{\varphi_{ii'}\}$, and for every $i\in I$
\[
X\ni(x_j)\mapsto\varphi_i((x_j))=x_i\in X_i.
\]

It can happen that for a projective family $\{X_i,\varphi_{ii'}\}_{i\in I}$ there are no nets compatible with the maps $\{\varphi_{ii'}\}$ even if all sets $\{X_i\}$ are non-empty and all maps $\{\varphi_{ii'}\}$ are surjective \cite{empty-0,empty}. Then $X$ is the empty set $\varnothing$ and each $\varphi_i:X=\varnothing\to X_i$ is an empty map. But then the empty set and the empty maps still satisfy the definition of the projective limit. Thus the projective limit of a projective family may be the empty set.

\subsection{Inductive limit of an inductive family of $C^*$-algebras}

Let $(I,\geq)$ be again a directed set. A family $\{\A_i,h_{i'i}\}_{i\in I}$ is an {\em inductive family} (or {\em a direct system}) of $C^*$-algebras if for every $i\in I$ $\A_i$ is a $C^*$-algebra, $\{h_{i'i}:\A_i\to \A_{i'}\}$ are injective $*$-homomorphisms of the algebras defined for every pair $i'\geq i$ and for every triplet $i''\geq i'\geq i$
\begin{equation}
h_{i''i'}\circ h_{i'i}=h_{i''i}.
\label{hhh}
\end{equation}

Let us define a relation on a disjoint union $\bigsqcup_{i\in I}\A_i$: we say that $a\in \A_i$ is in relation with $a'\in \A_{i'}$, $a\sim a'$, if for some $i''\geq i',i$
\[
h_{i''i}(a)=h_{i''i'}(a').
\]
It is easy to check that this is an equivalence relation. The following quotient space
\[
\A_0:=\Big(\bigsqcup_{i\in I}\A_i\Big)\Big/\!\!\sim
\]
will be called here an {\em algebraic inductive limit} of the family $\{\A_i,h_{i'i}\}_{i\in I}$. $\A_0$ is naturally a normed $*$-algebra: denoting by $[b]\in \A_0$ an equivalence class of $b\in\bigsqcup_{i\in I} \A_i$ one defines
\begin{align*}
z[a]+z'[a']&:=[zh_{i''i}(a)+z'h_{i''i'}(a')], & [a][a']&:=[h_{i''i}(a)h_{i''i'}(a')],\\
[a]^*&:=[a^*],& ||\,[a]\,||&:=||a||,
\end{align*}
where $a\in \A_i$, $a'\in \A_{i'}$, $z,z'\in \C$ and $i''\geq i',i$. These operations and the norm do not depend on the choice of $i''$ or the choice of representatives of the classes $[a]$, $[a']$  because the maps $\{h_{i'i}\}$ are injective $*$-homomorphisms which satisfy the condition \eqref{hhh} (an important fact here is that every {\em injective} $*$-homomorphism between $C^*$-algebras preserves the norms).

The completion
\[
\A:=\overline{\A_0}
\]
in the norm defined above is a $C^*$-algebra called an {\em inductive limit} (or a {\em direct limit}) of the inductive family  $\{\A_i,h_{i'i}\}_{i\in I}$.

There exists a natural injective $*$-homomorphism from every $\A_i$ into the inductive limit $\A$:
\begin{equation}
\A_i\ni a\mapsto [a]\in \A,
\label{hom}
\end{equation}
which allows us to treat each $\A_i$ as a $C^*$-subalgebra of $\A$.

\subsection{States on $C^*$-algebras}

Here we recall some basic facts concerning states on $C^*$-algebras \cite{brat-rob, araki}.

A linear functional $s$  on a $C^*$-algebra ${\cal A}$ valued in the complex numbers is positive if for every $a\in{\cal A}$  
\[
 s(a^*a)\geq 0.
\]
A positive linear functional $s$ of norm $||s||=1$ is a state on $\cal A$.  

\begin{thr}
Let $\cal A$ be a $C^*$-algebra. For every  $a\in\A$ there exists a state $s$ on $\A$ such that $s(a^*a)=||a||^2$.
\label{exists}
\end{thr}

\begin{thr}
Let ${\cal A}'$ be a $C^*$-subalgebra  of a $C^*$-algebra ${\cal A}$ and let $ s'$ be a state on ${\cal A}'$. Then there exists a state $ s$ on $\cal A$  such that its restriction to ${\cal A}'$  is the state $ s'$.
\label{ext}
\end{thr}

\begin{lm}
Let $\cal A$ be a $C^*$-algebra with a unit $\mathbf{1}$. Then for every state $s$ on $A$
\begin{equation*}
s(\mathbf{1})=1.
\end{equation*}
\label{s-1}
\end{lm}

Denote by $\B(\h)$ the $C^*$-algebra of all bounded operators on a {\em separable} Hilbert space $\h$. We say that $b\in\B(\h)$ is positive if for every $\psi\in\h$
\[
\scal{\psi}{b\psi}\geq 0
\]
(here $\scal{\cdot}{\cdot}$ is the scalar product on $\h$). Consider two operators $a,a'\in\B(\h)$. We say that $a'$ is larger than or equal to $a$, $a'\geq a$, if the difference $a'-a$ is a positive operator. A net $(a_i)$ labeled by elements of a directed set $I$ is a bounded increasing net if $i'\geq i$ implies $a_{i'}\geq a_i$ and the net $(||a_i||)$ of real numbers is bounded. An operator $b\in \B(\h)$ is an upper bound of a net $(a_i)$ if for every $i\in I$ $b\geq a_i$. If the net is increasing and bounded then there exist upper bounds for it. Then there also exists the least upper bound of the net i.e. an upper bound $c$ such that for every upper bound $b$ of the net $b\geq c$. Therefore the operator $c$ is denoted by ${\rm sup\,} (a_i)$.

We say that a state $\rho$ on $\B(\h)$ is {\em normal} if for every bounded increasing net $(a_i)$ of positive operators in $\B(\h)$
\[
\rho({\rm sup\,} (a_i))={\rm sup\,}\rho(a_i).
\]
There are states on $\B(\h)$ which are not normal. The set of all normal states is $*$-weak dense in the set of all states on $\B(\h)$ i.e. for every state $ s$ on $\B(\h)$ there exists a sequence $( \rho_n)$ of normal states on the algebra such that for every $a\in\B(\h)$
\[
\lim_{n\to\infty} \rho_n(a)= s(a).
\]

A {\em density operator} $\tilde{\rho}$ on $\h$ is a positive trace class operator of trace equal $1$. There is one-to-one correspondence between density operators on $\h$ and normal states on the algebra $\B(\h)$---each density operator $\tilde{\rho}$ defines such a state $\rho$ via the following formula
\[
\B(\h)\ni a\mapsto \rho(a):=\tr(a\tilde{\rho})\in \C,
\]
on the other hand for each normal state $\rho$ there exists a unique density operator $\tilde{\rho}$ such that the relation above is satisfied. Because of this correspondence in this paper we will not distinguish between density operators on $\h$ and normal states on $\B(\h)$.

\section{Projective constructions of spaces of quantum states for field theories}

In this section we will present the modified projective construction---if only this construction can be successfully applied to a field theory then it provides a ``sufficiently large'' space of quantum states. Before that we will describe some details of the (possibly flawed) hitherto construction which will be necessary for the presentation of the modified one.

In what follows we will apply the most advanced formulation of the hitherto construction elaborated by Lan\'ery and Thiemann in \cite{proj-lt-II} since it is best suited for our goal. Moreover, this formulation encompasses all earlier applications of the construction. However the reader should remember that the problem of ``too small'' spaces of projective quantum states is not specific merely to the Lan\'ery-Thiemann formulation but appears also in the earlier papers including the original one \cite{kpt}.

\subsection{Family of factorized Hilbert spaces}

In order to obtain a projective family $\{\D_\la,\pi_{\la\la'}\}_{\la\in\Lambda}$ from a directed set $\Lambda$ of finite physical systems defined over a phase space the set has to be chosen in a proper way. A criterion for such a choice can be formulated as follows \cite{proj-lt-II}: the set $\Lambda$ is chosen properly if the family $\{\h_\la\}_{\la\in\Lambda}$ of the Hilbert spaces associated with the systems is extendable to a richer structure  which will be called here a {\em  family of factorized Hilbert spaces}:

\begin{df}
A family of factorized Hilbert spaces is a quintuplet
\[
\Big(\Lambda,\h_\lambda,\tilde{\h}_{\lambda'\lambda},\Phi_{\lambda'\lambda},\Phi_{\lambda''\lambda'\lambda}\Big)
\]
such that:
\begin{enumerate}
\item $\Lambda$ is a directed set,
\item for every $\lambda\in\Lambda$ $\h_\lambda$ is a separable Hilbert space,
\item for every $\lambda'\geq \lambda$ $\tilde{\h}_{\lambda'\lambda}$ is a Hilbert space, and
\begin{equation}
\Phi_{\lambda'\lambda}:\h_{\lambda'}\to\tilde{\h}_{\lambda'\lambda}\ot\h_\lambda
\label{Phi}
\end{equation}
a Hilbert space isomorphism (for other pairs $(\lambda',\lambda)$ $\tilde{\h}_{\lambda'\lambda} $ and $\Phi_{\lambda'\lambda}$ are not defined); moreover $\dim\h_{\lambda\lambda}=1$ and $\Phi_{\lambda\lambda}$ is trivial\footnote{Assume that $\h'$ is a one dimensional Hilbert space.  A Hilbert space isomorphism $\Phi:\h\to\h'\ot\h$ is trivial if there exists a normed element $e$ of $\h'$ such that $\Phi(\psi)=e\ot\psi$ for every $\psi\in\h$.}.
\item for every $\lambda''\geq\lambda'\geq \lambda$
\[
\Phi_{\lambda''\lambda'\lambda}:\tilde{\h}_{\lambda''\lambda}\to\tilde{\h}_{\lambda''\lambda'}\ot\tilde{\h}_{\lambda'\lambda}
\]
is a Hilbert space isomorphism such that the following diagram
\[
\begin{CD}
\h_{\lambda''} @>\Phi_{\lambda''\lambda}>> \tilde{\h}_{\lambda''\lambda}\ot\h_\lambda\\
@VV\Phi_{\lambda''\lambda'}V       @VV\Phi_{\lambda'' \lambda'\lambda}\ot\id V \\
\tilde{\h}_{\lambda''\lambda'}\ot\h_{\lambda'}@>\id\ot\Phi_{\lambda'\lambda}>>  \tilde{\h}_{\lambda''\lambda'}\ot\tilde{\h}_{\lambda'\lambda}\ot\h_\lambda
\end{CD}
\]
is commutative (for other triplets $(\lambda'',\lambda',\lambda)$ $\Phi_{\lambda''\lambda'\lambda}$ are not defined); moreover if $\lambda''=\lambda'$ or $\lambda'=\lambda$ then $\Phi_{\lambda''\lambda'\lambda}$ is trivial. 
\end{enumerate}
\label{ffHs}
\end{df}
\noindent In fact, what may really be difficult in a construction of a family $\{\D_\la,\pi_{\la\la'}\}_{\la\in\Lambda}$ for a given phase space is the construction of the family of factorized Hilbert spaces---as we will see in Section \ref{hit}, once the latter construction is done the former one follows straightforwardly.

The notion of family of factorized Hilbert spaces was introduced in \cite{proj-lt-II} and the application of the projective construction to LQG with the $SU(2)$ symmetry presented in \cite{proj-lqg-I} is based on this notion. However, as we will show in Appendix \ref{ffHs-old}, spaces of projective quantum states constructed in the earlier papers \cite{kpt,q-nonl,q-stat} can also be seen as derived from families of factorized Hilbert spaces.

\begin{figure}[h]
\psfrag{h1}{$\h_\la$}
\psfrag{h2}{$\h_{\la'}$}
\psfrag{h3}{$\h_{\la''}$}
\psfrag{h4}{$\tilde{\h}_{\la'\la}$}
\psfrag{h5}{$\tilde{\h}_{\la''\la'}$}
\psfrag{h6}{$\tilde{\h}_{\la''\la}$}
\psfrag{l1}{$\la$}
\psfrag{l2}{$\la'$}
\psfrag{l3}{$\la''$}
\psfrag{dof}{d.o.f.}
\begin{center}
\includegraphics[scale=1.27]{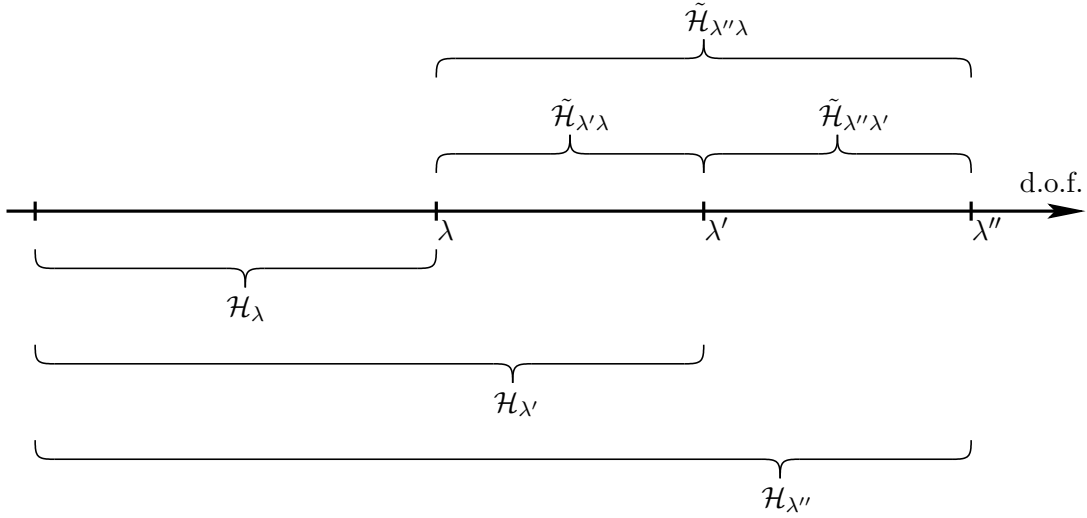}
\end{center}
\caption{A family of factorized Hilbert spaces---quantum d.o.f. of finite systems $\la''\geq\la'\geq\la$.}
\end{figure}

Let us emphasize that the notion defined by Definition \ref{ffHs} was called in \cite{proj-lt-II} a {\em projective system of quantum states spaces} but in our opinion this name is a bit misleading and imprecise. The term ``projective system'' used in \cite{proj-lt-II} is very similar to ``projective family'' (which is often called ``inverse system''), while the notion is certainly not very similar to a projective family. Moreover, ``quantum states space'' may mean a Hilbert space or a space of density operators. In fact, the name ``projective system of quantum states spaces'' suits much better every projective family $\{\D_\la,\pi_{\la\la'}\}_{\la\in\Lambda}$. This is why we decided to change the name---this new name does not suggest any similarity to the notion of projective family and makes it clear what sort of quantum state spaces is actually meant. It should also be clear that the participle ``factorized'' refers to the fact that each $\h_{\la'}$ is factorized as $\tilde{\h}_{\la'\la}\ot \h_\la$.

Let us finally note that there is a slight difference between the original definition in \cite{proj-lt-II} and the present one i.e. Definition \ref{ffHs}: unlike in \cite{proj-lt-II} we restrict ourselves to {\em separable} Hilbert spaces $\{\h_\la\}$ since they are Hilbert spaces of finite physical systems.

\subsection{The hitherto construction \label{hit}}

Assume that $\big(\Lambda,\h_\lambda,\tilde{\h}_{\lambda'\lambda},\Phi_{\lambda'\lambda},\Phi_{\lambda''\lambda'\lambda}\big)$ is a family of factorized Hilbert spaces. Let $\B_\lambda$ and $\tilde{\B}_{\lambda'\lambda}$ be the $C^*$-algebras of all bounded operators on, respectively, $\h_\lambda$ and $\tilde{\h}_{\lambda'\lambda}$. Denote by $\mathbf{1}_{\lambda'\lambda}$ the unit element of $\tilde{\B}_{\lambda'\lambda}$ and for every $\lambda'\geq\lambda$ define \cite{proj-lt-II}
\begin{equation}
\B_\lambda\ni a\mapsto \iota_{\lambda'\lambda}(a):=\Phi^{-1}_{\la'\la}\circ(\mathbf{1}_{\la'\la}\ot a)\circ\Phi_{\la'\la}\in\B_{\la'}.
\label{iota}
\end{equation}
Clearly, the map $\iota_{\la'\la}$ is a unital\footnote{A homomorphism from a unital algebra $\cal A$ to a unital algebra ${\cal A}'$  is unital if it maps the unit of $\cal A$ to the unit of ${\cal A}'$.} injective $*$-homomorphism. As a direct consequence of point 4 of Definition \ref{ffHs} for every triplet $\la''\geq\la'\geq\la$ \cite{proj-lt-II}
\begin{equation}
\iota_{\lambda''\lambda'}\circ\iota_{\lambda'\lambda}=\iota_{\lambda''\lambda}.
\label{ttt}
\end{equation}
Consequently, $\{\B_\la,\iota_{\la'\la}\}_{\la\in\Lambda}$ is an inductive family of $C^*$-algebras \cite{kpt}.

Let $\la'\geq\la$. Consider a normal state $\rho\in \D_{\la'}$ on $\B_{\la'}$. Then
\[
\B_\la\ni a\mapsto\rho(\iota_{\la'\la}(a))\in \C
\]
is a state on $\B_\la$. Taking into account the form of the map $\iota_{\la'\la}$ we see that the pull-back
\[
\rho\mapsto\iota^*_{\la'\la}(\rho)
\]
corresponds to $(i)$ mapping $\rho$ by means of $\Phi_{\la'\la}$ to a normal state on a $C^*$-algebra of all bounded operators on $\tilde{\h}_{\la'\la}\ot\h_\la$ (i.e. to a density operator on $\tilde{\h}_{\la'\la}\ot\h_\la$)  and then $(ii)$ to mapping the resulting state by means of the partial trace with respect to $\tilde{\h}_{\la'\la}$ to a normal state on $\B_{\la}$ (i.e. to a density operator on $\h_{\la}$). Thus the map
\begin{equation}
\pi_{\la\la'}:=\iota^*_{\la'\la}
\label{pi}
\end{equation}
is a surjection from $\D_{\la'}$ onto $\D_\la$. It follows from \eqref{ttt} that for every triplet $\la''\geq\la'\geq\la$
\[
\pi_{\la\la'}\circ\pi_{\la'\la''}=\pi_{\la\la''},
\]
which means that $\{\D_\la,\pi_{\la\la'}\}_{\la\in\Lambda}$ is a projective family.

As stated in the introduction the projective limit $\D$ of the family $\{\D_\la,\pi_{\la\la'}\}_{\la\in\Lambda}$ is meant to serve as  the space of quantum states for a field theory underlying the family of factorized Hilbert spaces and it is this space which may be possibly ``too small'' or empty. 

\subsection{The modified construction}

The point of departure for the modified construction  is the same family of factorized Hilbert spaces. Let us denote by $\S_\la$ the space of {\em all} states on $\B_\la$. For every $\la'\geq\la$ the map
\[
\Pi_{\la\la'}:=\iota^*_{\la'\la}
\]
is a map from $\S_{\la'}$ into $\S_\la$---using Theorem \ref{ext} one can show that the map is surjective. Again, by virtue of \eqref{ttt} for every triplet $\la''\geq\la'\geq\la$
\[
\Pi_{\la\la'}\circ\Pi_{\la'\la''}=\Pi_{\la\la''},
\]
which means that $\{\S_\la,\Pi_{\la\la'}\}_{\la\in\Lambda}$ is a projective family.

Since every space $\D_\la$ is a proper subset of $\S_\la$ and for every $\la'\geq\la$
\[
\Pi_{\la\la'}\big|_{\D_{\la'}}=\pi_{\la\la'}
\]
the new family $\{\S_\la,\Pi_{\la\la'}\}_{\la\in\Lambda}$ can be seen as a natural extension of the original family $\{\D_\la,\pi_{\la\la'}\}_{\la\in\Lambda}$.

\begin{lm}
The projective limit $\S$ of $\{S_\la,\Pi_{\la\la'}\}_{\la\in\Lambda}$ is non-empty.
\end{lm}
\begin{proof}
Denote by $\B$ a $C^*$-algebra defined as the inductive limit of the inductive family $\{\B_\la,\iota_{\la'\la}\}_{\la\in \Lambda}$ and by $\S_{\B}$ the set of all states on $\B$. By virtue of Theorem \ref{exists} the set $\S_{\B}$ is non-empty.

Let $s$ be a state on $\B$. The natural injective $*$-homomorphism \eqref{hom} from $\B_\la$ to $\B$ maps $a$ to its equivalence class $[a]$ and therefore 
\begin{equation}
\B_\la\ni a\mapsto  s_\la(a):= s([a])\in\C
\label{om-la}
\end{equation}
is a continuous positive functional on $\B_\la$ such that $||s_\la||\leq 1$.

Let $\mathbf{1}_\la$ be the unit of $\B_\la$. Since all homomorphisms $\{\iota_{\la'\la}\}$ are unital the equivalence class $[\mathbf{1}_\la]$ is the unit of both $\B_0$ and $\B$. Thus by virtue of Lemma \ref{s-1}
\[
|s_\la(\mathbf{1}_\la)|= |s([\mathbf{1}_\la])|=1 =||\mathbf{1}_\la||,
\]
hence\footnote{If $\cal A$ is a unital $C^*$-algebra then the norm of its unit is equal to $1$ \cite{brat-rob}.} $||s_\la||=1$ and $s_\la$ is a state on $\B_\la$ i.e. $s_\la\in \S_\la$.    

Suppose that $\la'\geq\la$ and $a\in\B_\la$. Then
\[
(\Pi_{\la\la'}( s_{\la'}))(a)= s_{\la'}(\iota_{\la'\la}(a))= s([\iota_{\la'\la}(a)])= s([a])= s_{\la}(a)
\]
---here we used the fact that according to the definition of the inductive limit $[\iota_{\la'\la}(a)]=[a]$. Consequently,
\[
\Pi_{\la\la'}( s_{\la'})= s_\la
\]
which means that the net $( s_\la) $ defined by $ s\in\S_{\B}$ is an element of the projective limit $\S$ of $\{S_\la,\Pi_{\la\la'}\}_{\la\in\Lambda}$.
\end{proof}

In fact, the set $\S$ is not only non-empty but its elements are in one-to-one correspondence with the states on $\B$:

\begin{lm}
The map
\begin{equation}
\S_{\B}\ni s\to\sigma(s):=(s_\la) \in\S,
\label{s-s}
\end{equation}
where $s_\la$ is given by \eqref{om-la}, is a bijection.
\end{lm}
\begin{proof}
We show first that the map \eqref{s-s} is surjective. Consider a net $(s_\la)\in\S$.  Let $b$ be an element of the algebraic inductive limit ${\cal B}_0$ of $\{\B_\la,\iota_{\la'\la}\}_{\la\in\Lambda}$. Then there exists $\la\in\Lambda$ and $a\in \B_\la$ such that $a$ is a representative of $b$ i.e. $b=[a]$. Define
\[
 s_0(b):= s_\la(a).
\]
It is clear that this definition does not depend on the choice of the representative $a$.

It is straightforward to check that $s_0$ is a linear functional on ${\cal B}_0$. Moreover
\begin{align*}
  s_0(b^*b)&=s_{\la}(a^*a)\geq 0,\\
  | s_0([b])|&=|s_\la(a)|\leq ||a||=||\,[b]\,||,
\end{align*}
which means that $ s_0$ can be unambiguously extended to a continuous positive linear functional $ s$ on $\B$ such that $||s||\leq 1$. 

Let $\mathbf{1}_\la$ be the unit of $\B_\la$. As we know already $[\mathbf{1}_\la]$ is the unit of both $\B_0$ and $\B$. Therefore
\[
|s([\mathbf{1}_\la])|=|s_0([\mathbf{1}_\la])|=|s_\la(\mathbf{1}_\la)|=1=||\,[\mathbf{1}_\la]\,||,
\]
where in the third step we used Lemma \ref{s-1}. Consequently, $||s||=1$ and  $s$ is a state on $\B$.  

Obviously, for every $\la\in \Lambda$ and for every  $a\in \B_\la$ 
\[
 s([a])= s_\la(a),
\]
which means that $\sigma(s)=(s_\la)$. Thus the map \eqref{s-s} is surjective.

Suppose that there exist $s,s'\in\S_{\B}$ such that  $\sigma(s)=\sigma(s')=(s_\la)\in \S$. Let $b$ be any element of $\B_0$ and let $a\in \B_\la$ be a representative of $b$. Then
\[
s(b)=s([a])=s_\la(a)=s'([a])=s'(b).
\]
Thus $s$ and $ s'$ coincide on $\B_0$ being a dense subset of $\B$. Hence $s=s'$ and the map \eqref{s-s} is injective.    
\end{proof}

By virtue of the lemma above we can identify the sets $\S_{\B}$ and $\S$.

Let us note finally that the projective limit $\D\subset \S=\S_{\B}$. Therefore $\D$ is non-empty if and only if there exists a state $ s$ on $\B$ such that for every $\la\in\Lambda$ the state $ s_\la$ on $\B_\la$ given by \eqref{om-la} is {\em normal}.

\subsection{Some remarks on the modified construction}

As stated in \cite{araki}
\begin{quote}
``(...) the theory which considers all states on\footnote{Here $\B(\h)$ denotes the $C^*$-algebra of all bounded operators on a separable Hilbert space $\h$.} $\B(\h)$ and the one which considers only the normal states are physically equivalent, and we can use either theory according to our convenience.''
\end{quote}
 It is clear that at the present moment we cannot state analogously that ``the quantum field theory which considers all states on $\B$ (i.e. elements of the projective limit $S$) and the one which considers only the states in $\D$ are physically equivalent'' because we do not know how ``large'' the limit $\D$ is in general. On the other hand we will argue that the space $\S$ is ``large enough'' to serve as the space of quantum states for the quantum field theory.

Let us recall that the idea of the projective construction of quantum states is to divide a (classical) field theory into many finite classical physical systems, to quantize each finite system separately and then ``to glue'' the resulting spaces of quantum states by means of the projective techniques into one meant to serve as the state space for the corresponding quantum field theory. As stated in the introduction the finite systems altogether should encompass all d.o.f. of the original phase space. Since each $\B_\la$ contains all observables of the system $\la$ it is reasonable to treat the inductive limit $\B$ of the family $\{\B_\la,\iota_{\la'\la}\}_{\la\in\Lambda}$ as the algebra of observables of the quantum theory \cite{kpt}. Since the limit $\S$ coincides with the set $\S_{\B}$ of all states on $\B$ it seems that it is ``large enough'' to be used as the state space  for the quantum theory.

Moreover, the limit $\S$ encompasses all quantum states of all quantized finite systems. Indeed, assume that $ s_\la$ is a state on $\B_\la$. This algebra is a $C^*$-subalgebra of $\B$ and therefore by virtue of Theorem \ref{ext} there exists a state $ s\in\S_{\B}=\S$ such that its restriction to $\B_\la$ coincides with $ s_\la$. Thus in particular, the space $\S$ encompasses all normal states used to construct the projective family $\{\D_\la,\pi_{\la\la'}\}_{\la\in\Lambda}$.

States on the algebra $\B_\la$ which are not normal are not easy to describe and handle in comparison to normal ones. Fortunately, for all practical purposes connected to experiments we can  work exclusively with normal states even if the limit $\D$ is ``too small'' or empty. The reason for this is the following. In every experiment we can measure only a {\em finite} number of observables $\{a_1,\ldots,a_n\}$  associated with the system $\la$. Assume that the quantum field is in a state $s\in\S$. Then the experiment can be described in terms of (expectation) values of the state $s$ on a {\em finite} number of operators $\{b_1,\ldots,b_m\}\subset \B_\la$ being functions\footnote{For example, $b_j$ may be the orthogonal  projection on a subspace of $\h_\la$ spanned by a finite number of eigenvectors of $a_k$.} of $\{a_1,\ldots,a_n\}$. For each $b_j$  
\[
s([b_j])=s_\la(b_j),
\]
where $[b_j]\in \B$ and $s_\la\in \S_\la$ is an element of (the net) $s$. Suppose that $s_\la$ is not normal. Since the set $\D_\la$ is $*$-weak dense in $\S_\la$ for every $\epsilon>0$ there exists $\rho_\la\in \D_\la$ such that
\[
|s_\la(b_j)-\rho_\la(b_j)|<\epsilon
\]
for every $j\in\{1,\ldots,m\}$. Thus we can always find a normal state $\rho_\la$ which cannot be experimentally distinguished from the state $s_\la$ since there are no experiments of perfect accuracy.

All these facts mean that the question whether $\D$ is non-empty or ``large enough'' is now not very relevant---by using the algebraic states instead of merely the normal ones we obtained the physically valid space $\S$ of quantum states for the quantum field theory, but still for all practical purposes we can use exclusively the normal states.

\subsection{Remarks on a vacuum state}

In case of a physical system with infinitely many degrees of freedom, the observable algebra $\B$ and its space of states $\S= \S_B$ can describe physical situations which go far beyond the framework of quantum field theory. Indeed, the latter corresponds to the {\em vacuum sector} of the pair $(\B , \S)$ and can be obtained {\em via} the Gelfand-Naimark-Segal (GNS) construction. This construction also provides the final Hilbert space of the theory.

For this purpose we need a vacuum state $V \in \S$. Here, we briefly sketch  the idea on how to construct this state (cf. \cite{kpt, kpt-qed}).

As stated before, we treat each of the ``truncated'' theories on the level $\lambda \in \Lambda$ as an approximation of the full theory. In particular, each of them should be equipped with a  ``truncated'' Hamiltonian operator. Denote by $v_{\lambda^\prime} \in \D_{\lambda^\prime} \subset \S_{\lambda^\prime}$  the ground state of this operator and by $v_{\lambda  \lambda^\prime}\in S_{\la}$ its projection on an arbitrary subsystem $\lambda \le \lambda^\prime$. The entire difficulty in the construction of the quantum field theory consists in proving the existence of the limit (in the *-weak topology)
\[
V_{\lambda} = \lim_{\lambda^\prime} v_{\lambda  \lambda^\prime} \in \S_\lambda.
\]
If the limit $V_\la$ exists for every $\lambda$, then the net $(V_\lambda) $ is automatically compatible with the projections $\{\Pi_{\la\la'}\}$ which means that $(V_\lambda)  $ belongs to the projective limit $\S$ and thereby it defines a state $V$ on the algebra $\B$. Being the limit of the approximate vacuum states, it is a natural candidate for a non-perturbative vacuum of the continuum theory and the starting point for the GNS construction of the Hilbert
space of its quantum states.

We stress that each $v_\lambda \in \D_\lambda$ is normal because it is a pure state corresponding to the ground state of the ``truncated'' Hamiltonian operator which approximates the full Hamiltonian of the theory. However, the limit $V_{\lambda}$ (if there is any) does not need to be normal. This is where the construction presented above enters the game. 

The existence of the above limit of vacuum states may be extremely difficult
to prove. A realistic attitude consists, therefore,
in a detailed analysis of the theory on the level of its lattice approximations: if the continuum limit of these approximations
does exist, the numerical results obtained on the level of a
``sufficiently late'' approximation $\lambda$ should provide a ``sufficiently good'' approximation of the true
physical theory.

\section{Summary}

In this paper we removed the possible flaw of the space $\D$ meant to serve as the state space for a quantum field theory---this space is defined as a projective limit and therefore it may be empty or ``too small'' to serve this purpose. This flaw was removed by a natural extension of the projective family $\{\D_\la,\pi_{\la\la'}\}_{\la\in\Lambda}$ defining the limit $\D$---each space $\D_\la$ of normal states  on the algebra $\B_\la$ of quantum observables of the finite system $\la$ was extended to the space $\S_\la$ of all algebraic states on the algebra. The resulting projective family $\{\S_\la,\Pi_{\la\la'}\}_{\la\in\Lambda}$ possesses the non-empty projective limit $\S$. Moreover, the limit $\S$ contains the limit $\D$ as its proper subset and coincides with the set of all algebraic states on the algebra $\B$ of observables of the quantum field theory.

\paragraph{Acknowledgment} We are very grateful to Prof. Stanis{\l}aw L. Woronowicz, Pawe{\l} Kasprzak and Piotr So{\l}tan for valuable discussions and help. This research was supported by Narodowe Centrum Nauki, Poland (grant 2016/21/B/ST1/00940).

\appendix

\section{A family of factorized Hilbert spaces in \cite{q-nonl} \label{ffHs-old}}

In \cite{q-nonl} the construction of the projective family $\{{\cal D}_\lambda,\pi_{\lambda\lambda'}\}_{\la\in\Lambda}$ is based on a family of Hilbert spaces $\{\h_\lambda,\tilde{\h}_{\lambda'\lambda}\}_{\la\in\Lambda}$---each ${\cal D}_\lambda$ is a space of density operators on $\h_\lambda$, and for each pair $\lambda'\geq\lambda$ the projection $\pi_{\lambda\lambda'}:{\cal D}_{\lambda'}\to\D_\lambda$ is defined as a partial trace with respect to the Hilbert space $\tilde{\h}_{\lambda'\lambda}$. In this appendix we will show that the set $\{\h_\lambda,\tilde{\h}_{\lambda'\lambda}\}_{\la\in\Lambda}$ defined in \cite{q-nonl} form in a natural way a family of factorized Hilbert spaces and that this family provides the same projective family $\{{\cal D}_\lambda,\pi_{\lambda\lambda'}\}_{\la\in\Lambda}$.

On the other hand, the constructions described in \cite{kpt} and \cite{q-stat} are particular examples of the construction presented in \cite{q-nonl}. Thus the result to be achieved here means that the families $\{{\cal D}_\lambda,\pi_{\lambda\lambda'}\}_{\la\in\Lambda}$ obtained in these three papers can be seen as originating from families of factorized Hilbert spaces.

\subsection{Preliminaries}

Here we present some notions and facts defined and proven in \cite{q-nonl}.

In \cite{q-nonl} the directed set $\Lambda$ is a subset of a Cartesian product of two directed sets $\hat{\mathbf{F}}$ and $\mathbf{K}$. An element $\hat{F}\in \hat{\mathbf{F}}$ describes a finite number of momentum d.o.f. and an element $K\in\mathbf{K}$ does a finite number of configurational d.o.f. of a field theory. Thus each $\lambda\equiv(\hat{F},K)\in \Lambda$ describes a finite physical system.

Given system $\lambda\equiv(\hat{F},K)$, its configurational space is a (finite dimensional) linear space denoted by $Q_K$. For every pair $\lambda'\equiv(\hat{F}',K')\geq\lambda\equiv(\hat{F},K)$ there exist a linear surjection
\[
\pr_{KK'}:Q_{K'}\to Q_K
\]
and a linear injection
\[
\omega_{\lambda'\lambda}: Q_K\to Q_{K'}
\]
such that
\begin{align}
\pr_{KK'}\circ\omega_{\lambda'\lambda}&=\id,\label{po-id}\\
Q_{K'}&=\ker\pr_{KK'}\oplus \omega_{\lambda'\lambda}(Q_K).\label{QK'-dec}
\end{align}
If for some $K$ and $K'$ $Q_K=Q_{K'}$ then $\pr_{KK'}=\id$.

For every triplet $\lambda''\equiv(\hat{F}'',K'')\geq\lambda'\equiv(\hat{F}',K')\geq\lambda\equiv(\hat{F},K)$
\begin{align}
\pr_{KK''}&=\pr_{KK'}\circ\pr_{K'K''},\label{ppp}\\
\omega_{\lambda''\lambda}&=\omega_{\lambda''\lambda'}\circ\omega_{\lambda'\lambda}.
\label{ooo}
\end{align}
The spaces $Q_{K''}$ and $\ker\pr_{KK''}$ can be decomposed in the following way
\begin{align}
Q_{K''}&= \ker\pr_{K'K''}\oplus\omega_{\lambda''\lambda'}(\ker\pr_{KK'})\oplus\omega_{\lambda''\lambda}(Q_K),\label{QK''-dec}\\
\ker\pr_{KK''}&=\ker\pr_{K'K''}\oplus\omega_{\lambda''\lambda'}(\ker\pr_{KK'}).\label{ker-dec}
\end{align}

For each $\lambda$ its configurational space $Q_K$ is equipped with a measure $d\mu_\lambda$. For every pair $\lambda'\geq\lambda$ the space $\ker\pr_{KK'}$ is equipped with a measure $d\tilde{\mu}_{\lambda'\lambda}$ such that
\begin{equation}
d\mu_{\lambda'}=d\tilde{\mu}_{\lambda'\lambda}\times \omega_{\lambda'\lambda*}(d\mu_\lambda).
\label{mmm}
\end{equation}
For every triplet $\lambda''\geq\lambda'\geq\lambda$
\begin{align}
d\mu_{\lambda''}&=d\tilde{\mu}_{\lambda''\lambda'}\times\omega_{\lambda''\lambda'*}(d\tilde{\mu}_{\lambda'\lambda})\times\omega_{\lambda''\lambda'*}(d\mu_{\lambda}),\label{m1}\\
d\tilde{\mu}_{\lambda''\lambda'}&=d\tilde{\mu}_{\lambda''\lambda'}\times\omega_{\lambda''\lambda'*}(d\tilde{\mu}_{\lambda'\lambda})\label{m2}.
\end{align}

A Hilbert space of a system $\lambda$ is defined as a space of square integrable functions on $Q_K$:
\[
\h_\lambda:=L^2(Q_K,d\mu_\lambda).
\]
Moreover, for every $\lambda'\geq\lambda$
\[
\tilde{\h}_{\lambda'\lambda}:=L^2(\ker\pr_{KK'},d\tilde{\mu}_{\lambda'\lambda}).
\]
If $\ker\pr_{KK'}=\{0\}\in Q_{K'}$ then
\[
\int_{\ker\pr_{KK'}}f \,d\tilde{\mu}_{\lambda'\lambda}:=f(0)\,\xi\in\C
\]
for some real number $\xi>0$ independent of $f$. Then $\tilde{\h}_{\lambda'\lambda}$ is naturally isomorphic to the set $\C$ of complex numbers equipped with a scalar product
\[
\scal{z}{z'}:=\bar{z}z'\xi.
\]

\subsection{Construction of the family of factorized Hilbert spaces}

Let $\lambda'\equiv (\hat{F}',K')\geq\lambda\equiv (\hat{F},K)$. By virtue of injectivity of $\omega_{\lambda'\lambda}$ and the decomposition \eqref{QK'-dec} the following map
\begin{equation}
\ker\pr_{KK'}\times Q_K\ni(q',q)\mapsto\phi_{\la'\la}(q',q):=q'+\omega_{\la'\la}(q)\in Q_{K'}
\label{phi-ll}
\end{equation}
is bijective. It follows from \eqref{mmm} that
\[
d\mu_{\lambda'}=\phi_{\lambda'\lambda*}(d\tilde{\mu}_{\lambda'\lambda}\times d\mu_\lambda).
\]
Therefore
\begin{equation}
\Phi_{\lambda'\lambda}:=\phi^*_{\lambda'\lambda}:\h_{\lambda'}\to\tilde{\h}_{\lambda'\lambda}\ot \h_\lambda
\label{Phi-ll}
\end{equation}
is a Hilbert space isomorphism.

Assume now that $\lambda'\equiv(\hat{F}',K')=\lambda\equiv(\hat{F},K)$. Then $\pr_{KK}=\id$ and  consequently $\ker\pr_{KK}=\{0\}$. This fact together with \eqref{po-id} mean that $\omega_{\lambda\lambda}=\id$. Therefore 
\[
\phi_{\la\la}(0,q)=0+q=q,
\]
hence for every function $\psi$ on $Q_K$  
\[
(\phi^*_{\la\la}\psi)(0,q)=\psi(q)=\tilde{1}(0)\,\psi(q),
\]
where $\tilde{1}$ is a function on $\ker\pr_{KK}$ of a value equal $1$. Thus
\[
\phi^*_{\la\la}\psi=\tilde{1}\,\psi
\]
Therefore for every $\psi\in\h_\la$ 
\begin{equation}
\Phi_{\lambda\lambda}(\psi)=\tilde{1}\ot\psi.
\label{Phi-1}
\end{equation}
Now Equation \eqref{mmm} reads
\[
d\mu_{\lambda}=d\tilde{\mu}_{\lambda\lambda}\times d\mu_{\lambda}
\]
---it follows from it that the measure $d\tilde{\mu}_{\lambda\lambda}$ is given by $\xi=1$. Therefore $\tilde{1}$ is a normed element of a one-dimensional Hilbert space $\h_{\lambda\lambda}$. This fact and \eqref{Phi-1} mean that $\Phi_{\lambda\lambda}$ is trivial.

The decomposition \eqref{QK''-dec} allows us to define on  $\ker \pr_{K'K''}\times \ker \pr_{KK'}\times Q_K$ the following two maps valued in $Q_{K''}$:
\begin{align*}
&(q'',q',q)&&\mapsto&&(q''+\omega_{\lambda''\lambda'}(q'),q)&&\mapsto&&q''+\omega_{\lambda''\lambda'}(q')+\omega_{\lambda''\lambda}(q),\\
&(q'',q',q)&&\mapsto&&(q'',q'+\omega_{\lambda'\lambda}(q))&&\mapsto&&q''+\omega_{\lambda''\lambda'}(q')+\omega_{\lambda''\lambda'}(\omega_{\lambda'\lambda}(q)).
\end{align*}
 The identity \eqref{ooo} guarantees that the maps coincide. This fact together with \eqref{ker-dec} is equivalent to the commutativity of the following diagram:
\[
\begin{CD}
Q_{K''} @<\phi_{\lambda''\lambda}< < \ker{\pr_{KK''}}\times Q_K\\
@AA\phi_{\lambda''\lambda'}A       @AA\phi_{\lambda''\lambda'\la}\times\id A \\
\ker\pr_{K'K''}\times Q_{K'}@< \id\times\phi_{\lambda'\lambda}< <  \ker \pr_{K'K''}\times \ker \pr_{KK'}\times Q_K
\end{CD},
\]
where $\phi_{\la''\la'\la}$ is a map defined as follows: 
\[
\ker \pr_{K'K''}\times \ker \pr_{KK'}\ni(q'',q')\mapsto\phi_{\la''\la'\la}(q'',q'):=q''+\omega_{\la''\la'}(q')\in \ker\pr_{KK''}.
\] 
This map is a bijection by virtue of the decomposition \eqref{ker-dec} and of injectivity of $\omega_{\la''\la'}$. Thus all the maps appearing in the diagram are bijections.

The diagram above and the properties \eqref{mmm}, \eqref{m1} and \eqref{m2} allow us to construct another commutative diagram:
\[
\begin{CD}
d\mu_{\lambda''} @<\phi_{\lambda''\lambda *}< < d\tilde{\mu}_{\lambda''\lambda}\times d\mu_\lambda\\
@AA\phi_{\lambda''\lambda'*}A       @AA\phi_{\la''\la'\la *}\times\id_* A \\
d\tilde{\mu}_{\lambda''\lambda'}\times d\mu_{\lambda'}@< \id_*\times \phi_{\lambda'\lambda *}< <  d\tilde{\mu}_{\lambda''\lambda'}\times d\tilde{\mu}_{\lambda'\lambda}\times d\mu_\lambda
\end{CD}\,\,.
\]

It follows from the two diagrams above that the following one
\[
\begin{CD}
\h_{\lambda''} @>\Phi_{\lambda''\lambda}>> \tilde{\h}_{\lambda''\lambda}\ot\h_\lambda\\
@VV\Phi_{\lambda''\lambda'}V       @VV\phi^*_{\lambda'' \lambda'\lambda}\ot\id V \\
\tilde{\h}_{\lambda''\lambda'}\ot\h_{\lambda'}@>\id\ot\Phi_{\lambda'\lambda}>>  \tilde{\h}_{\lambda''\lambda'}\ot\tilde{\h}_{\lambda'\lambda}\ot\h_\lambda
\end{CD}
\]
is commutative also. It is clear that $\phi^*_{\lambda'' \lambda'\lambda}:\tilde{\h}_{\lambda''\lambda}\to \tilde{\h}_{\lambda''\lambda'}\ot\tilde{\h}_{\lambda'\lambda}$ is a Hilbert space isomorphism. Therefore we define
\begin{equation}
\Phi_{\lambda'' \lambda'\lambda}:=\phi^*_{\lambda'' \lambda'\lambda}.
\label{Phi-lll}
\end{equation}

Suppose now that $\lambda'\equiv(\hat{F}',K')=\lambda\equiv(\hat{F},K)$. Then $\ker\pr_{KK}=\{0\}$ and
\[
\phi_{\lambda''\lambda\lambda}(q'',0)=q''+\omega_{\la''\la}(0)=q''.
\]
Consequently, for every function $\psi$ on $\ker\pr_{KK''}$  
\[
(\phi^*_{\lambda''\lambda\lambda}\psi)(q'',0)=\psi(q'')=\psi(q'')\,\tilde{1}(0),
\]
where now $\tilde{1}$ is a function on $\ker\pr_{KK}$ of a value equal $1$. Thus
\[
\phi^*_{\lambda''\lambda\lambda}\psi=\psi\,\tilde{1}.
\]
Since $d\tilde{\mu}_{\lambda\lambda}$ is given by $\xi=1$ the function $\tilde{1}$ is a normed element of one-dimensional Hilbert space $\tilde{\h}_{\la\la}$. Thus
\[
\Phi_{\lambda''\lambda\lambda}(\psi)=\psi\ot \tilde{1},
\]
which means that the map is trivial.

Assume now that $\lambda''\equiv(\hat{F}'',K'')=\lambda'\equiv(\hat{F}',K')$. Then $\ker\pr_{K'K'}=\{0\}$ and $\omega_{\la'\la'}=\id$ hence 
\[
\phi_{\lambda'\lambda'\lambda}(0,q')=0+\omega_{\la'\la'}(q')=q'.
\]
Consequently, for every function $\psi$ on $\ker\pr_{KK'}$  
\[
(\phi^*_{\lambda'\lambda'\lambda}\psi)(0,q')=\psi(q')=\tilde{1}(0)\,\psi(q'),
\]
where now $\tilde{1}$ is a function on $\ker\pr_{K'K'}$ of a value equal $1$. Thus
\[
\phi^*_{\lambda'\lambda'\lambda}\psi=\tilde{1}\,\psi.
\]
Since the function $\tilde{1}$ is a normed element of one-dimensional Hilbert space $\tilde{\h}_{\la'\la'}$ the map $\Phi_{\lambda'\lambda'\lambda}$ is trivial.

We conclude that the family $\{\h_\lambda,\tilde{\h}_{\lambda'\lambda}\}_{\la\in\Lambda}$ of Hilbert spaces defined in \cite{q-nonl} equipped with the maps \eqref{Phi-ll} and \eqref{Phi-lll} form a family of factorized Hilbert spaces.

Let us show finally that this family of factorized Hilbert spaces provides the same projective family $\{\D_\la,\pi_{\la\la'}\}_{\la\in\Lambda}$ as in \cite{q-nonl}. To reach the goal it is enough to prove that projections $\{\pi_{\la\la'}\}$ defined by Equations \eqref{pi}, \eqref{iota}, \eqref{Phi-ll} and \eqref{phi-ll}  coincide with those defined in \cite{q-nonl}.    

Consider a pair $\lambda'\equiv(\hat{F}',K')\geq\lambda\equiv(\hat{F},K)$. As stated in Section \ref{hit} the projection $\pi_{\la\la'}$ defined by \eqref{pi} and \eqref{iota} maps a density operator $\rho$ on $\h_{\la'}$ to a density operator on $\tilde{\h}_{\la'\la}\ot\h_\la$: 
\[
\rho\mapsto \Phi_{\la'\la}\circ\rho\circ\Phi_{\la'\la}^{-1}
\]
and then evaluates the partial trace with respect to $\tilde{\h}_{\la'\la}$:
\[
\Phi_{\la'\la}\circ\rho\circ\Phi_{\la'\la}^{-1}\mapsto\tr_{\tilde{\h}_{\la'\la}}\big(\Phi_{\la'\la}\circ\rho\circ\Phi_{\la'\la}^{-1}\big) 
\]     
Let 
\[
\h_{\la'\la}:=L^2(\omega_{\la'\la}(Q_K),\omega_{\la'\la*}d\mu_\la).
\]
As shown in \cite{q-nonl}
\begin{equation}
\h_{\la'}=\tilde{\h}_{\la'\la}\ot\h_{\la'\la}.
\label{fact}
\end{equation}
If $\Phi_{\la'\la}$ is given by \eqref{Phi-ll} and \eqref{phi-ll} then
\[
\Phi_{\la'\la}=\id\ot\omega^*_{\la'\la},
\]
where the above tensor product is defined with respect to the factorization \eqref{fact} and $\omega^*_{\la'\la}$ is understood here as an isomorphism from $\h_{\la'\la}$ to $\h_\la$. Thus
\[
\Phi_{\la'\la}(\h_{\la'})=\tilde{\h}_{\la'\la}\ot\omega^*_{\la'\la}(\h_{\la'\la})=\tilde{\h}_{\la'\la}\ot\h_{\la}.
\]

Therefore the projection $\pi_{\la\la'}$ can be seen equivalently as a map which first evaluates the partial trace of $\rho$ with respect to $\tilde{\h}_{\la'\la}$ and the factorization \eqref{fact}:    
\[
\rho\mapsto \tr_{\tilde{\h}_{\la'\la}}\rho,
\]
and then maps the resulting density operator on $\h_{\la'\la}$ to one on $\h_{\la}$ by means of the isomorphism $\omega^*_{\la'\la}$:
\[
\tr_{\tilde{\h}_{\la'\la}}\rho\mapsto \omega^*_{\la'\la}\circ\big(\tr_{\tilde{\h}_{\la'\la}}\rho\big)\circ \omega^{*-1}_{\la'\la}.
\]   
This is exactly how the projection $\pi_{\la\la'}$ is defined in \cite{q-nonl}.


\end{document}